\documentclass[11pt]{article}
\usepackage{amsmath,amssymb,amsthm}
\usepackage[english]{babel}

\usepackage[width=\textwidth,indention=5mm,justification=centerlast,
 font=small,labelfont=bf,labelsep=period]{caption}

\usepackage[pdftex]{graphicx}
\usepackage[pdftex,dvipsnames]{color}

\usepackage[pdfpagemode=UseOutlines,
 bookmarks=true,bookmarksopen=true,bookmarksnumbered=true,
 pdffitwindow=true,pdfpagelayout=SinglePage,pdfhighlight=/P,
 colorlinks=true,linkcolor=blue,citecolor=blue,urlcolor=blue,
 unicode=true]{hyperref}

\newcommand{\bbC}{\mathbb{C}}
\newcommand{\bbZ}{\mathbb{Z}}
\newcommand{\bbR}{\mathbb{R}}

\newcommand{\cL}{\mathcal{L}}

\newcommand{\Integer}{\mathbb Z}
\newcommand{\Real}{\mathbb R}
\newcommand{\Complex}{\mathbb C}
\newcommand{\I}{{\rm i\, }}
\renewcommand{\Re}{\mathop{\rm Re}}
\renewcommand{\Im}{\mathop{\rm Im}}

\theoremstyle{plain}
\newtheorem{theorem}{Theorem}
\newtheorem{lemma}[theorem]{Lemma}
\newtheorem{proposition}[theorem]{Proposition}

\newtheorem{definition}{Definition}

\theoremstyle{definition}

\newtheorem{remark}{Remark}

\title{Discrete pluriharmonic functions as solutions of linear pluri-Lagrangian systems}
\author{A.I. Bobenko\thanks{Institut
f\"ur Mathematik, Technische Universit\"at Berlin, Strasse~des 17.~Juni 136,
10623 Berlin, Germany}\and Yu.B. Suris$^*$}
\date{\today}

\begin{document}
\maketitle

\begin{abstract}
Pluri-Lagrangian systems are variational systems with the multi-dimensional consistency property.
This notion has its roots in the theory of pluriharmonic functions, in the Z-invariant models of statistical mechanics, in the theory of variational symmetries going back to Noether and in the theory of discrete integrable systems.
 A $d$-dimensional pluri-Lagrangian problem can be described
as follows: given a $d$-form $L$ on an $m$-dimensional space, $m > d$,
whose coefficients depend on a function $u$ of $m$ independent variables
(called field), find those fields $u$ which deliver critical points to the action functionals
$S_\Sigma=\int_\Sigma L$ for any $d$-dimensional manifold $\Sigma$ in the $m$-dimensional space. We investigate discrete 2-dimensional linear pluri-Lagrangian systems, i.e. those with quadratic Lagrangians $L$. The action is a discrete analogue of the Dirichlet energy, and solutions are called discrete pluriharmonic functions. We classify linear pluri-Lagrangian systems with Lagrangians depending on diagonals. They are described by generalizations of the star-triangle map. Examples of more general quadratic Lagrangians are also considered.
\medskip

\medskip

\end{abstract}

\section{Introduction}\label{s:intro}

In the last decade, a new understanding of integrability of discrete systems as their multi-dimensional consistency has been a major breakthrough \cite{BS1}, \cite{N}. This led to classification of discrete 2-dimensional integrable systems (ABS list)  \cite{ABS}, which turned out to be rather influential. According to the concept of multi-dimensional consistency, integrable two-dimensional systems can be imposed in a consistent way on all two-dimensional sublattices of a lattice $\bbZ^m$ of arbitrary dimension. This means that the resulting multi-dimensional system possesses solutions whose restrictions to any two-dimensional sublattice are generic solutions of the corresponding two-dimensional system. To put this idea differently, one can impose the two-dimensional equations on any quad-surface in $\bbZ^m$ (i.e., a surface composed of elementary squares), and transfer solutions from one such surface to another one, if they are related by a sequence of local moves, each one involving one three-dimensional cube, like the moves shown of Fig. \ref{Fig: local moves}.

\begin{figure}[htbp]
\begin{center}
\includegraphics[width=0.5\textwidth]{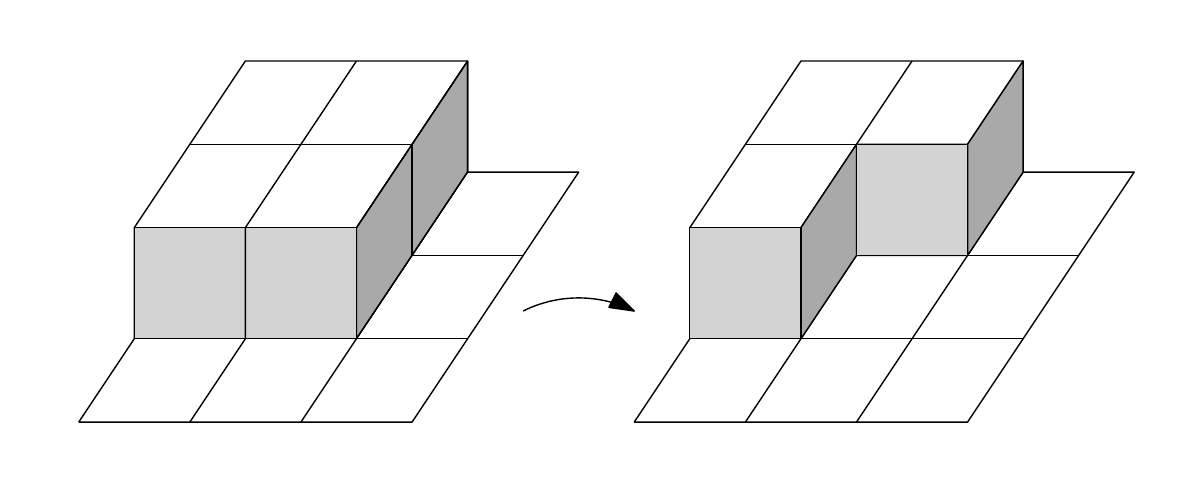}
\caption{Local move of a quad-surface involving one three-dimensional cube}
\label{Fig: local moves}
\end{center}
\end{figure}

A further fundamental conceptual development was initiated by Lobb and Nijhoff \cite{LN1} and further generalized in various directions in \cite{LN2, LNQ, XNL, BS2, BollSuris2}. This develpoment deals with variational (Lagrangian) formulation of discrete multi-dimensionally consistent systems.
Its main idea can be summarized as follows: solutions of integrable systems deliver critical points simultaneously for actions along all possible manifolds of the corresponding dimension in multi-time; the Lagrangian form is closed on solutions. This idea is, doubtless, rather inventive (not to say exotic) in the framework of the classical calculus of variations. However,  it has significant precursors. These are:
\begin{itemize}
\item Theory of pluriharmonic functions and, more generally, of pluriharmonic maps \cite{R, OV, BFPP}. By definition, a pluriharmonic function of several complex variables $f:\bbC^m\to\bbR$ minimizes the Dirichlet functional $E_\Gamma=\int_\Gamma |(f\circ \Gamma)_z|^2dz\wedge d\bar z$ along any holomorphic curve in its domain $\Gamma:\bbC\to\bbC^m$. Differential equations governing pluriharmonic functions (and maps) are heavily overdetermined. Therefore it is not surprising that they belong to the theory of integrable systems.
\item Baxter's Z-invariance of solvable models of statistical mechanics  \cite{Bax1, Bax2}. This concept is based on invariance of the partition function of solvable models under elementary local transformations of the underlying planar graph. It is well known (see, e.g., \cite{BMS}) that one can associate the planar graphs underling these models with quad-surfaces in $\bbZ^m$. On the other hand, the classical mechanical analogue of the partition function is the action functional. This makes the relation of Z-invariance to the concept of closedness of the Lagrangian 2-form rather natural, at least at the heuristic level. Moreover, this relation has been made mathematically precise for a number of models, through the quasiclassical limit, in the work of Bazhanov, Mangazeev, and Sergeev \cite{BMS1, BMS2}.
\item The classical notion of variational symmetries,  going back to the seminal work of E.~Noether \cite{Noether}, turns out to be directly related to the idea of the closedness of the Lagrangian form in the multi-time. This was further elucidated in \cite{S2}.
\end{itemize}

Especially the relation with the pluriharmonic functions motivates a novel term we have introduced to describe the situation we are interested in, namely: given a $d$-form $\cL$ in the $m$-dimensional space ($d<m$), depending on a function $u$ of $m$ variables, one looks for functions $u$ which deliver critical points to actions $S_\Sigma=\int_\Sigma \cL$ corresponding to any $d$-dimensional manifold $\Sigma$. We call this a  {\em pluri-Lagrangian problem} and claim that integrability of variational systems should be understood as the existence of the pluri-Lagrangian structure. We envisage that this notion will play a very important role in the future development of the theory of integrable systems.

A general theory of one-dimensional pluri-Lagrangian systems has been developed in \cite{S}. It was demonstrated that for $d=1$ the pluri-Lagrangian property is characteristic for commutativity of Hamiltonian flows in the continuous time case and of symplectic maps in the discrete time case. This property yields that the exterior derivative of the multi-time Lagrangian 1-form is constant, $d\cL={\rm const}$. Vanishing of this constant (i.e., closedness of the Lagrangian 1-form) was shown to be related to integrability in the following, more strict, sense:
\begin{itemize}
\item[--] in the continuous time case, $d\cL=0$ is equivalent for the Hamiltonians of the commuting flows to be in involution,

\item[--] in the discrete time case, for one-parameter families of commuting symplectic maps, $d\cL=0$ is equivalent to the {\em spectrality property} \cite{KS}, which says that the derivative of the Lagrangian with respect to the parameter of the family is a generating function of common integrals of motion for the whole family.
\end{itemize}

A general theory of discrete two-dimensional pluri-Lagrangian systems has been developed in \cite{BPS2}. The main building blocks of the multi-time Euler-Lagrange equations for a discrete pluri-Lagrangian problem with $d=2$ were identified. These are the so called 3D-corner equations. The notion of consistency of the system of 3D-corner equations was discussed, and this system was analyzed for a special class of three-point 2-forms, corresponding to integrable quad-equations of the ABS list. On this way, a conceptual gap of the work by Lobb and Nijhoff was closed by showing that the corresponding 2-forms are closed not only on solutions of (non-variational) quad-equations, but also on general solutions of the corresponding multi-time Euler-Lagrange equations.

In the present paper, we study the case of quadratic Lagrangian 2-forms, leading to linear 3D-corner equations.

\section{Discrete pluri-Lagrangian problem and discrete pluriharmonic functions}

The following notations are used throughout the paper. The independent discrete
variables on the lattice $\Integer^N$ are denoted $n=(n_1,n_2,\dots,n_N)$. The lattice shifts are denoted by
$T_i:n\to n+ e_i$, where $e_i$ is the $i$-th coordinate vector.
\begin{itemize}
\item
Single superscript $i$ means that an object is associated with the edge $(n,n+e_i)$
and double superscript $ij$ is used for objects associated with the plaquette
$\sigma^{ij}=(n,n+e_i,n+e_i+e_j,n+e_j)$.
\item
Subscripts  denote shifts on the lattice. For instance, the fields at the corners of the plaquette
$\sigma^{ij}$ are denoted by $u, u_i, u_{ij}, u_j$.
\end{itemize}

In the most general form, consistent variational equations can be described as
follows.

\begin{definition}\label{def:pluriLagr problem} {\bf (2D pluri-Lagrangian problem)}

\begin{itemize}
\item Let $\cL$ be a discrete 2-form, i.e., a real-valued function of oriented elementary squares
\[
\sigma^{ij}=\left(n,n+e_{i},n+e_{i}+e_{j},n+e_{j}\right)
\]
of $\bbZ^m$, such that $\cL\left(\sigma^{ij}\right)=-\cL\left(\sigma^{ji}\right)$. We will assume that $\cL$ depends on some field assigned to the points of $\bbZ^m$, that is, on some $u:\bbZ^m\to \bbC$. More precisely, $\cL(\sigma^{ij})$ depends on the values of $u$ at the four vertices of $\sigma^{ij}$:
\begin{equation}\label{L-v}
 \cL(\sigma^{ij})=L(u,u_i,u_j,u_{ij}).
\end{equation}

\item To an arbitrary oriented quad-surface $\Sigma$ in $\bbZ^m$, there corresponds the {\em action functional}, which assigns to $u|_{V(\Sigma)}$, i.e., to the fields at the vertices of the surface $\Sigma$, the number
\begin{equation}\label{action on surface}
S_\Sigma=\sum_{\sigma^{ij}\in\Sigma}\cL(\sigma^{ij}).
\end{equation}
\item We say that the field $u:V(\Sigma)\to \bbC$ is a critical point of $S_\Sigma$, if at any interior point $n\in V(\Sigma)$, we have
\begin{equation}\label{eq: dEL gen}
    \frac{\partial S_\Sigma}{\partial u(n)}=0.
\end{equation}
Equations (\ref{eq: dEL gen}) are called {\em discrete Euler-Lagrange equations} for the action $S_\Sigma$ (or just for the surface $\Sigma$, if it is clear which 2-form $\cL$ we are speaking about).
\item We say that the field $u:\bbZ^m\to\bbC$ solves the {\em pluri-Lagrangian problem} for the Lagrangian 2-form $\cL$ if, {\em for any quad-surface $\Sigma$ in $\bbZ^m$}, the restriction $u|_{V(\Sigma)}$ is a critical point of the corresponding action $S_\Sigma$.
\end{itemize}
\end{definition}

\begin{definition}\label{def:pluriLagr system} {\bf (System of corner equations)}
A 3D-corner is a quad-surface consisting of three elementary squares adjacent to a vertex of valence 3.
The {\em system of corner equations} for a given discrete 2-form $\cL$ consists of discrete Euler-Lagrange equations for all possible 3D-corners in $\bbZ^m$. If the action for the surface of an oriented elementary cube $\sigma^{ijk}$ of the coordinate directions $i,j,k$ (which can be identified with the discrete exterior derivative $d\cL$ evaluated at $\sigma^{ijk}$) is denoted by ($T_k$ is the shift in the $k$-th coordinate direction)
\begin{equation}\label{eq: Sijk}
S^{ijk}=d\cL(\sigma^{ijk})=(T_k-I)\cL(\sigma^{ij})+(T_i-I)\cL(\sigma^{jk})+(T_j-I)\cL(\sigma^{ki}),
\end{equation}
then the system of corner equations consists of the eight equations
\begin{equation}\label{eq: corner eqs}
\begin{array}{llll}
\dfrac{\partial S^{ijk}}{\partial u}=0, & \dfrac{\partial S^{ijk}}{\partial u_i}=0, & \dfrac{\partial S^{ijk}}{\partial u_j}=0, & \dfrac{\partial S^{ijk}}{\partial u_k}=0, \\
\\
\dfrac{\partial S^{ijk}}{\partial u_{ij}}=0, & \dfrac{\partial S^{ijk}}{\partial u_{jk}}=0, & \dfrac{\partial S^{ijk}}{\partial u_{ik}}=0, & \dfrac{\partial S^{ijk}}{\partial u_{ijk}}=0
\end{array}
\end{equation}
for each triple $i,j,k$. Symbolically, this can be put as $\delta(d\cL)=0$, where $\delta$ stands for the ``vertical'' differential (differential with respect to the dependent field variables $u$).
\end{definition}

\begin{remark} {\bf(Corner equations: maps and correspondences)}
Each of eight corner equations (\ref{eq: corner eqs}) relates seven fields and may be not uniquely solvable with respect to them, i.e. generically define correspondences. This leads to a substantial complication of analysis of consistency of the system discussed in the rest of this section (non-connectedness of the variety of solutions). We do not refer to these problems here and assume that the variety of solutions is connected (or one can single out a ``physical'' connected component). In the case of linear pluri-Lagrangian systems considered in the present paper all corner equations are linear and these problems do not appear.
 \end{remark}

As demonstrated in \cite{BPS2}, the flower of any interior vertex of an oriented quad-surface in $\bbZ^m$ can be represented as a sum of (oriented) 3D-corners in $\bbZ^{m+1}$. Thus, the system of corner equations encompasses all possible discrete Euler-Lagrange equations for all possible quad-surfaces. In other words, solutions of a pluri-Lagrangian problem as introduced in Definition \ref{def:pluriLagr problem} are precisely solutions of the corresponding system of corner equations.
\medskip

\begin{definition}\label{def:L-consistency'} {\bf (Consistency of a pluri-Lagrangian problem)}
A pluri-Lagrangian problem with the 2-form $\cL$ is called {\em consistent} if, for any quad-surface $\Sigma$ flippable to the whole of $\bbZ^N$, any generic solution Euler-Lagrange equations for $\Sigma$ can be extended to a solution of the system of corner equations on the whole $\bbZ^N$.
\end{definition}

Like in the case of quad-equations \cite{ABS}, consistency is actually a local issue to be addressed for one elementary cube. The system of corner equations (\ref{eq: corner eqs}) for one elementary cube is heavily overdetermined. It consists of eight equations, each one connecting seven fields out of eight. Any six fields can serve as independent data, then one can use two of the corner equations to compute the remaining two fields, and the remaining six corner equations have to be satisfied identically. This justifies the following definition.

\begin{definition}\label{def: corner eqs consist} {\bf (Consistency of corner equations)} System (\ref{eq: corner eqs}) is called {\em consistent}, if it has the minimal possible rank 2, i.e., if exactly two of these equations are independent.
\end{definition}

The main feature of this definition is that the ``almost closedness'' of the 2-form $\cL$ on solutions of the system of corner equations is, so to say, built-in from the outset. One should compare the proof of the following theorem with similar proofs in \cite{BS2, S}.

\begin{theorem}\label{Th: almost closed}
For any triple of the coordinate directions $i,j,k$, the action $S^{ijk}$ over an elementary cube of these coordinate directions is constant on solutions of the system of corner equations (\ref{eq: corner eqs}):
\[
S^{ijk}(u,\ldots,u_{ijk})=c^{ijk}={\rm const}  \pmod{\partial S^{ijk}/\partial u=0,\ \ldots,\
\partial S^{ijk}/\partial u_{ijk}=0}.
\]
\end{theorem}
\begin{proof} On the connected six-dimensional manifold of solutions, the gradient of $S^{ijk}$ considered as a function of eight variables, vanishes by virtue of (\ref{eq: corner eqs}).
\end{proof}

\begin{definition} \label{def: pluriharmonic} {\bf (Discrete pluriharmonic functions)}

If all $\cL(\sigma^{ij})=L(u,u_i,u_j,u_{ij})$ are quadratic forms of their arguments, then the action functional $S_\Sigma$ is called the {\em Dirichlet energy} corresponding to the quad-surface $\Sigma$. We call a solution $u:\bbZ^N\to\bbC$ of the pluri-Lagrangian problem in this case {\em a discrete pluriharmonic function}.
\end{definition}

We are aware that our Definition \ref{def: pluriharmonic} is not an immediate discretization of the classical notion of pluriharmonic functions on $\bbC^m$, and might be therefore misleading. However, we believe that these two notions are close in spirit and we hope that the further developments will establish a closer relation between the classical and the discrete pluriharmonicity, including approximation theorems for harmonic functions through discrete harmonic functions, cf. \cite{ChS}.

In the context of Definition \ref{def: corner eqs consist}, functional independence is replaced by linear independence, and the rank is understood in the sense of linear algebra.

\begin{theorem}\label{Th: closed}
The 2-form $\cL$ is closed on pluriharmonic functions. Let $u:\bbZ^N\to\bbR$ be a discrete pluriharmonic function, and let $\Sigma$, $\widetilde\Sigma$ be two quad-surfaces with the same boundary, then the Dirichlet energies of harmonic functions $u|_\Sigma$ and $u|_{\widetilde\Sigma}$ coincide.
\end{theorem}
\begin{proof} The constants from Theorem \ref{Th: almost closed} can be evaluated on the trivial pluriharmonic function $u\equiv 0$, which gives $c^{ijk}=0$.
\end{proof}

\section{Example: discrete complex analysis}\label{s:H}

We recall general scheme of discrete complex analysis following \cite{M_CMP, M_complex} (see also \cite{BMS, ChS, Ke, Skopenkov}).
Let $D$ be a quad-graph. Consider a generic face with the vertices $0,1,2,3$ ordered cyclically counterclockwise, and let $d$, $d^*$ be the diagonals $(02)$, $(13)$, respectively. Assume that the weights $p\in\Complex$ are defined on diagonals of all quads, such that
\[
p(d^*)=\frac{1}{p(d)}.
\]

\begin{definition}\label{d.dharmonic} {\bf (Discrete holomorphic and discrete harmonic functions)}

\begin{itemize}
\item A function on the vertices of a quad-graph $f:V(D)\to\Complex$ is called {\em discrete holomorphic} if on every face of $D$ it satisfies
\begin{equation}\label{dCR}
f_2-f_0=\I p(d)(f_1-f_3).
\end{equation}
\item A function $u:V(D)\to\Real$ is called {\em discrete harmonic} if it is a real part of a discrete holomorphic function, $u=\Re f$.
\end{itemize}
\end{definition}

One may think about planar quads with the vertex coordinates $z_0,z_1,z_2,z_3$ in the complex plane and the weights defined as $p(d)=-\I\dfrac{z_2-z_0}{z_1-z_3}$.

\begin{theorem}
A function $u$ is discrete harmonic if and only if it satisfies $\Delta u=0$
with the following discrete Laplace operator
\begin{equation}\label{Delta u}
 (\Delta u)_0=\sum_{\text{{\rm quads} $(0,1,2,3)$} \atop \text{{\rm incident to vertex} $0$}}
  \Big(\frac{1}{\Re p}(u_2-u_0)+\frac{\Im p}{\Re p}(u_1-u_3)\Big),
\end{equation}
or, equivalently, if it is critical for the Dirichlet energy
\begin{equation}\label{DEnergy}
S_D(u)=\frac{1}{2}\sum_{F(D)}
   \Bigl(\frac{1}{\Re p}(u_2-u_0)^2+2\frac{\Im p}{\Re p}(u_2-u_0)(u_1-u_3)+
   \frac{|p|^2}{\Re p}(u_1-u_3)^2\Bigr)
\end{equation}
(where $p=p(d)=p(02)$).
\end{theorem}
Note that
\[
\frac{|p|^2}{\Re p}=\frac{1}{\Re p^*}, \quad \frac{\Im p}{\Re p}=-\frac{\Im p^*}{\Re p^*},\quad
\frac{1}{\Re p}=\frac{|p^*|^2}{\Re p^*},
\]
where $p^*=p(d^*)=p(13)=1/p$, which shows that the Dirichlet energy is well defined.

\begin{lemma}\label{l.conjugate}
Let $u:V(D)\to\Real$ be a discrete harmonic function on a simply connected
quad-graph $D$. Then there exists a function $v:V(D)\to\bbR$ such that $f=u+\I v:V(D)\to\Complex$ is
a a discrete holomorphic function. Such a function $v$ is unique up to an additive imaginary bi-constant (function taking one constant value on all black vertices and another constant value on all white vertices of $D$), is harmonic with the same weights and is called conjugate to $u$.
\end{lemma}
\begin{proof}
The conjugate function on a quad is defined by the formulas
\begin{equation}\label{dconjugate}
\begin{aligned}
 v_3-v_1 & =\frac{1}{\Re p}(u_2-u_0)+\frac{\Im p}{\Re p}(u_1-u_3),\\
 v_2-v_0 & =\frac{|p|^2}{\Re p}(u_1-u_3)+\frac{\Im p}{\Re p}(u_2-u_0).
\end{aligned}
\end{equation}
Closedness of the so defined 1-form for $v$ is equivalent to harmonicity of the function $u$ (see (\ref{Delta u})). In turn, Cauchy--Riemann equation (\ref{dCR}) for $f=u+\I v$ is equivalent to (\ref{dconjugate}).
\end{proof}

Recall (see, e.g., \cite{BMS}) that the quad-graph $D$ can be interpreted as a quad-surface in $\bbZ^N$. We are interested in extending all objects from $D$ to $\bbZ^N$. Suppose that the complex weights $p^{ij}$ are assigned to the diagonals $(n,n+e_i+e_j)$ of elementary squares $\sigma^{ij}$ of $\bbZ^N$, so that the weights $1/p^{ij}$ are assigned to the diagonals $(n+e_i,n+e_j)$.
\begin{definition}
\begin{itemize}
\item A function $f:\Integer^N\to\Complex$ is called discrete holomorphic if it
satisfies equation
\begin{equation}\label{dCR lattice}
f_{ij}-f=\I p^{ij}(f_i-f_j)
\end{equation}
on every elementary square of $\Integer^N$.
\item A function $u:\bbZ^N\to\bbR$ is called discrete pluriharmonic if its satisfies the discrete Laplace equation $\Delta u=0$, i.e., if it is critical for the Dirichlet energy
\begin{equation}\label{DEnergy lattice}
S_\Sigma(u)=\frac{1}{2}\sum_{\sigma^{ij}\in\Sigma}
   \Bigl(\frac{1}{\Re p^{ij}}(u_{ij}-u)^2+2\frac{\Im p^{ij}}{\Re p^{ij}}(u_{ij}-u)(u_i-u_j)+\frac{|p^{ij}|^2}{\Re p^{ij}}(u_i-u_j)^2\Bigr)
\end{equation}
on every quad-surface in $\Integer^N$.
\end{itemize}
\end{definition}

Like in Lemma \ref{l.conjugate}, for a discrete pluriharmonic function $u$, there exists a conjugate pluriharmonic function $v$ such that $f=u+\I v$ is a discrete holomorphic function on $\bbZ^N$. It is defined by
the formulas analogous to (\ref{dconjugate}):
\begin{equation}\label{dconjugate lattice}
\begin{aligned}
 v_j-v_i &= \frac{1}{\Re p^{ij}}(u_{ij}-u)+\frac{\Im p^{ij}}{\Re p^{ij}}(u_i-u_j),\\
 v_{ij}-v &= \frac{|p^{ij}|^2}{\Re p^{ij}}(u_i-u_j)+\frac{\Im p^{ij}}{\Re p^{ij}}(u_{ij}-u).
\end{aligned}
\end{equation}

Existence of discrete pluriharmonic functions is a condition on the weights $p^{ij}$. It can be locally reformulated as a condition on a 3D cube.

\begin{theorem}
Complex weights $p_{ij}$ correspond to a consistent discrete Laplace operator if and only if they satisfy the star-triangle equation
\begin{equation}\label{star-triangle}
 p^{ij}_k:=T_k p^{ij}=\frac{p^{ij}}{p^{ij}p^{jk}+p^{jk}p^{ki}+p^{ki}p^{ij}}.
\end{equation}
\end{theorem}
\begin{proof}
This follows from the consistency of discrete Cauchy-Riemann (called also discrete Moutard) equations (see \cite{DDG}).
\end{proof}

One can see that if projections of a quad-surface on all two-dimensional coordinate planes are injective then this quad-surface can be chosen as a support of initial data for the map (\ref{star-triangle}). In this case the weights on a quad-surface can be chosen arbitrarily and extended to the whole of $\bbZ^N$ via the map (\ref{star-triangle}). Thus, any discrete Laplace operator on such a quad-surface is integrable in the sense that it can be extended to a consistent Laplace operator on the whole lattice.

\section{Classification of discrete pluriharmonic functions with Lagrangians depending on diagonals}

\subsection{From discrete pluriharmonic functions to vector Moutard equations}
Consider Lagrangians
\[
L^{ij}=\frac{1}{2}\alpha^{ij}(u_{ij}-u)^2+\beta^{ij}(u_{ij}-u)(u_i-u_j)+
\frac{1}{2}\gamma^{ij}(u_i-u_j)^2.
\]
Let $u:\mathbb Z^3\to \mathbb C$ be a pluriharmonic function (i.e., let it satisfy all eight corner equations for each cube).
\begin{lemma}\label{Lemma conj harm}
Equations
\begin{eqnarray*}
v_{i}-v_{j} & = & -\frac{\partial L^{ij}}{\partial u}=\frac{\partial L^{ij}}{\partial u_{ij}}=\alpha^{ij}(u_{ij}-u)+\beta^{ij}(u_i-u_j)\\
v_{ij}-v & = & -\frac{\partial L^{ij}}{\partial u_i} = \frac{\partial L^{ij}}{\partial u_j} =
-\beta^{ij}(u_{ij}-u)-\gamma^{ij} (u_i-u_j)
\end{eqnarray*}
are consistent (by virtue of the corner equations for $u$) and define the function $v:\mathbb Z^3\to\mathbb C$, called {\em conjugate pluriharmonic function}, up to a bi-constant (constant on black and on white points separately).
\end{lemma}
\begin{proof}
Closedness of the so defined 1-form for $v$ is equivalent to the corner equations for the function $u$.
\end{proof}

The above equations can be put as
\begin{eqnarray*}
u_{ij}-u & = & \frac{1}{\alpha^{ij}}(v_i-v_j)-\frac{\beta^{ij}}{\alpha^{ij}}(u_i-u_j),\\
v_{ij}-v & = & \Big(-\gamma^{ij}+\frac{(\beta^{ij})^2}{\alpha^{ij}}\Big)(u_i-u_j)-
\frac{\beta^{ij}}{\alpha^{ij}}(v_i-v_j),
\end{eqnarray*}
or, in the matrix form,
\begin{equation}\label{matr Moutard}
\begin{pmatrix}u_{ij}-u \\ v_{ij}-v\end{pmatrix}=A^{ij}\begin{pmatrix} u_i-u_j \\ v_i-v_j \end{pmatrix}
\end{equation}
(vector Moutard equation), with the matrix coefficients
\begin{equation}\label{matr A}
    A^{ij}=\begin{pmatrix} b^{ij} & a^{ij} \\ c^{ij} & b^{ij} \end{pmatrix},
\end{equation}
whose entries are related to the coefficients of the Lagrangians $L^{ij}$ via
\begin{equation}\label{eq: A thru L}
b^{ij}=-\frac{\beta^{ij}}{\alpha^{ij}},\quad a^{ij}=\frac{1}{\alpha^{ij}}, \quad c^{ij}=\frac{(\beta^{ij})^2-\alpha^{ij}\gamma^{ij}}{\alpha^{ij}}.
\end{equation}

We say that vector Moutard equation (\ref{matr Moutard}) is consistent if arbitrary initial data $u,v$ and $u_i,v_i$ ($i=1,2,3$) can be extended to the fields $u,v$ at all vertices of a 3D cube so that (\ref{matr Moutard}) is fulfilled on all six faces. Valid initial data for a consistent vector Moutard equation consist of 8 numbers, e.g.,
\begin{itemize}
\item the values $u,v$ and all $u_i,v_i$ for $i=1,2,3$.
\end{itemize}
These data can be put into the one-to-one correspondence with the following initial data for a pair of conjugate pluriharmonic functions $(u,v)$ within one elementary cube:
\begin{itemize}
\item values of $u$ at 6 vertices, which define by corner equations the values of $u$ at the remaining two vertices, and the values of $v$ at two vertices (one white and one black).
\end{itemize}
Thus, the task of the classification of consistent pluri-Lagrangian systems with quadratic Lagrangians depending on diagonals turns out to be equivalent to the task of the description of consistent systems of vector Moutard equations.

\subsection{Vector Moutard equations and non-commutative star-triangle relation}

The solution of the latter task is described in the following theorem.
\begin{theorem}\label{th star-triang classif}
Matrices $A^{ij}$ of the form (\ref{matr A}) assigned to plaquettes $\sigma^{ij}(n)=(n,n+e_i,n+e_i+e_j,n+e_j)$ in $\mathbb Z^N$ serve as coefficients of a consistent system of vector Moutard equations if and only if their entries satisfy
\begin{equation}\label{main cond Moutard}
\lambda a^{ij}+\mu(-1)^{|n|}b^{ij}+\nu c^{ij}=0
\end{equation}
for some fixed triple $(\lambda,\mu,\nu)$, where $|n|=n_1+\ldots+n_m$, and their evolution is expressed (in the generic case, when $\lambda,\nu\neq 0$) through a solution of {\em coupled star-triangle relations}
\begin{equation}\label{coupled star triang}
\frac{1}{p^{ij}_k}=\frac{\lambda}{\nu}\cdot\frac{q^{ij}q^{jk}+q^{jk}q^{ki}+q^{ki}q^{ij}}{q^{ij}},
\quad
\frac{1}{q^{ij}_k}=\frac{\lambda}{\nu}\cdot\frac{p^{ij}p^{jk}+p^{jk}p^{ki}+p^{ki}p^{ij}}{p^{ij}},
\end{equation}
via the following relations:
\begin{equation}\label{pq}
p^{ij}=a^{ij}+\xi b^{ij},\quad q^{ij}=a^{ij}+\eta b^{ij},
\end{equation}
and
\begin{equation}\label{pq hat}
p^{ij}_k=a^{ij}_k-\xi b^{ij}_k,\quad q^{ij}_k=a^{ij}_k-\eta b^{ij}_k,
\end{equation}
where $-\xi$, $-\eta$ are the two roots of the quadratic equation $\lambda \xi^2+\mu\xi+\nu=0$.
\end{theorem}
We will prove this theorem by considering one 3D cube. It will be convenient to simplify the notations: we will denote the matrix coefficients assigned to the faces of the cube by $A_i$, $\hat{A_i}$, where, for instance,
\[
A_1=A^{23}, \quad \hat{A}_1=A^{23}_1,
\]
and other ones obtained by cyclic shifts of indices (the entries will carry the same indices).
\begin{proposition}
Vector Moutard equation is consistent if and only if the matrix coefficients $A_i$ satisfy
\begin{equation}\label{matr necessary}
A_1A_2^{-1}A_3= A_3A_2^{-1}A_1,
\end{equation}
and then the matrices $\hat{A}_i$ are given by {\em non-commutative star-triangle relations}
\begin{equation}\label{matr star triang}
-\hat{A}_i^{-1} = A_j+A_k+A_kA_i^{-1}A_j,
\end{equation}
where $(i,j,k)$ is any cyclic permutation of $(1,2,3)$. Note that the cyclic permutation of indices in (\ref{matr necessary}) leads to an equivalent equation.
\end{proposition}
\begin{proof}
Moutard equations on the three faces adjacent to the vertex 123 yield three different linear expressions for $(u_{123},v_{123})$ in terms of the initial data. Identifying the corresponding coefficients, we obtain:
\begin{eqnarray}\label{star triang aux}
-\hat{A}_3A_2=-\hat{A}_2A_3=I+\hat{A}_1(A_2+A_3),\nonumber\\
-\hat{A}_1A_3=-\hat{A}_3A_1=I+\hat{A}_2(A_3+A_1),\\
-\hat{A}_2A_1=-\hat{A}_1A_2=I+\hat{A}_3(A_1+A_2).\nonumber
\end{eqnarray}
The first equations in each of these three lines allow us to express, say, $\hat{A}_1$ and $\hat{A}_3$ linearly through $\hat{A}_2$ in two different ways, e.g.,
\[
\hat{A}_1=\hat{A}_2A_1A_2^{-1}=\hat{A}_2A_3A_2^{-1}A_1A_3^{-1}.
\]
These two expressions agree if condition (\ref{matr necessary}) is satisfied. Now the second equations in each of the three lines (\ref{star triang aux}) yield linear equations for each of $\hat{A}_i$, which result in (\ref{matr star triang}).
\end{proof}

For matrices of the form $A=\begin{pmatrix} b & a \\ c & b \end{pmatrix}$, as in (\ref{matr A}), we will give a convenient resolution of the constraint (\ref{matr necessary}), which will allow us to give a complete solution of the system consisting of  (\ref{matr necessary}), (\ref{matr star triang}) on $\bbZ^N$.
\begin{lemma}
Matrix equation (\ref{matr necessary}) is equivalent to
\begin{equation}\label{det}
\left|\begin{array}{ccc} a_1 & b_1 & c_1 \\
                         a_2 & b_2 & c_2 \\
                         a_3 & b_3 & c_3
\end{array}\right|=0,
\end{equation}
i.e., there exist $\lambda,\mu,\nu\in\bbR$ such that
\begin{equation}\label{lin dep}
\lambda a_i+\mu b_i+ \nu c_i=0, \quad i=1,2,3.
\end{equation}
\end{lemma}
\begin{proof}
We have:
\[
\begin{pmatrix}b_1 & a_1 \\ c_1 & b_1\end{pmatrix}\begin{pmatrix} b_2 & -a_2 \\ -c_2 & b_2\end{pmatrix}\begin{pmatrix} b_3 & a_3 \\ c_3 & b_3 \end{pmatrix}
\]
\[
=\begin{pmatrix} b_1b_2b_3-a_1c_2b_3+a_1b_2c_3-b_1a_2c_3 & b_1b_2a_3-a_1c_2a_3+a_1b_2b_3-b_1a_2b_3 \\
c_1b_2b_3-b_1c_2b_3+b_1b_2c_3-c_1a_2c_3 & b_1b_2b_3-c_1a_2b_3+c_1b_2a_3-b_1c_2a_3\end{pmatrix}.
\]
This has to be equal to the same matrix with $1\leftrightarrow 3$. A direct inspection shows that the off-diagonal terms give no conditions, while the diagonal ones give only one condition:
\[
-a_1c_2b_3+a_1b_2c_3-b_1a_2c_3=-b_1c_2a_3+c_1b_2a_3-c_1a_2c_3,
\]
which is nothing but (\ref{det}).
\end{proof}

{\em Proof of Theorem \ref{th star-triang classif}.}
Formula (\ref{main cond Moutard}) describes propagation of condition (\ref{lin dep}) to the lattice $\bbZ^N$. To show that this condition propagates under flipping the 3D cubes, consider an elementary cube with (\ref{lin dep}) satisfied on three faces adjacent to one vertex. We have to show that for the opposite three faces there follows
\begin{equation}\label{hat lin dep}
\lambda \hat{a}_i-\mu \hat{b}_i+ \nu \hat{c}_i=0, \quad i=1,2,3.
\end{equation}
From
\[
-\hat{A}_2^{-1}=A_1+A_3+A_1A_2^{-1}A_3
\]
we derive:
\begin{eqnarray*}
\hat{a}_2 & \sim & (a_1+a_3)(b_2^2-a_2c_2)+b_1b_2a_3-a_1c_2a_3+a_1b_2b_3-b_1a_2b_3,\\
\hat{c}_2 & \sim & (c_1+c_3)(b_2^2-a_2c_2)+c_1b_2b_3-b_1c_2b_3+b_1b_2c_3-c_1a_2c_3,\\
-\hat{b}_2 & \sim & (b_1+b_3)(b_2^2-a_2c_2)+b_1b_2b_3-a_1c_2b_3+a_1b_2c_3-b_1a_2c_3.
\end{eqnarray*}
We further compute:
\begin{eqnarray*}
\lambda\hat{a}_2-\mu\hat{b}_2+\nu\hat{c}_2 & \sim &
b_1b_2(\lambda a_3+\mu b_3+\nu c_3)-a_1c_2(\lambda a_3+\mu b_3+\nu c_3)\\
 & & +b_2b_3(\lambda a_1+\mu b_1+\nu c_1)-a_2c_3(\lambda a_1+\mu b_1+\nu c_1)\\
 & & +a_1c_3(\lambda a_2+\mu b_2+\nu c_2)-b_1b_3(\lambda a_2+\mu b_2+\nu c_2) =0
\end{eqnarray*}
(we added and subtracted $\lambda a_1a_2c_3$, $\mu b_1b_2b_3$, and $\nu a_1c_2c_3$). Analogously, $\lambda\hat{a}_1-\mu\hat{b}_1+\nu\hat{c}_1=\lambda\hat{a}_3-\mu\hat{b}_3+\nu\hat{c}_3=0$.

To parametrize matrix formulas (\ref{matr star triang}) by a system of coupled scalar star-triangle relations, we need a more precise computation. We start with
\begin{equation}\label{matr star triang elem}
    \frac{1}{\hat{\Delta}}\begin{pmatrix} \hat{a}_2 \\ \hat{c}_2 \\ -\hat{b}_2 \end{pmatrix}=
    \begin{pmatrix} a_1+a_3 \\ c_1+c_3 \\ b_1+b_3\end{pmatrix}+\frac{1}{\Delta}
    \begin{pmatrix} b_1b_2a_3-a_1c_2a_3+a_1b_2b_3-b_1a_2b_3 \\
                    c_1b_2b_3-b_1c_2b_3+b_1b_2c_3-c_1a_2c_3 \\
                    b_1b_2b_3-a_1c_2b_3+a_1b_2c_3-b_1a_2c_3
    \end{pmatrix},
\end{equation}
where $\Delta=b_2^2-a_2c_2$ and $\hat{\Delta}=\hat{b}_2^2-\hat{a}_2\hat{c}_2$. Upon condition (\ref{lin dep}), we find:
\begin{equation}\label{Delta}
\Delta=\frac{\lambda a_2^2+\mu a_2b_2+\nu b_2^2}{\nu}=\frac{\lambda}{\nu}(a_2+\xi b_2)(a_2+\eta b_2),
\end{equation}
and, similarly,
\begin{equation}\label{Delta hat}
\hat{\Delta}=\frac{\lambda}{\nu}(\hat{a}_2-\xi\hat{b}_2)(\hat{a}_2-\eta\hat{b}_2).
\end{equation}
Now the latter formula can be put as
\[
\frac{1}{\hat{a}_2-\eta\hat{b}_2}=\frac{\lambda}{\nu}\frac{\hat{a}_2-\xi\hat{b}_2}{\hat{\Delta}},
\]
and upon using the expressions from (\ref{matr star triang elem}), we put this as
\begin{eqnarray}
&& \frac{1}{\hat{a}_2-\eta\hat{b}_2} = \frac{\lambda}{\nu}\Big(a_1+a_3+\xi(b_1+b_3)\Big)
\nonumber\\
&& \ +\frac{\lambda/\nu}{\Delta}\Big(b_1b_2(a_3+\xi b_3)-a_1c_2(a_3+\xi b_3)+a_1b_2(b_3+\xi c_3)-b_1a_2(b_3+\xi c_3)\Big).\qquad\label{proof aux}
\end{eqnarray}
By transforming the right-hand side of the latter equation, one systematically eliminates $c_i$ in favor of $a_i$ and $b_i$, according to (\ref{lin dep}). In particular, one easily finds that
$b_3+\xi c_3=-\frac{\lambda\xi}{\nu}(a_3+\xi b_3)$. Thus, we find:
\begin{eqnarray*}
 && b_1b_2(a_3+\xi b_3)-a_1c_2(a_3+\xi b_3)+a_1b_2(b_3+\xi c_3)-b_1a_2(b_3+\xi c_3)  \\
 && \quad =\Big(b_1b_2+a_1\dfrac{\lambda a_2+\mu b_2}{\nu}\Big)(a_3+\xi b_3)-\dfrac{\lambda\xi}{\nu}(a_1b_2-b_1a_2)(a_3+\xi b_3)  \\
 && \quad =\dfrac{1}{\nu}\Big(\lambda a_1a_2+(\mu-\lambda\xi)a_1b_2+\lambda\xi b_1a_2+\nu b_1b_2\Big)(a_3+\xi b_3)\\
 && \quad =\dfrac{\lambda}{\nu}(a_1+\xi b_1)(a_2+\eta b_2)(a_3+\xi b_3).
\end{eqnarray*}
In the last transformation, we have taken into account that $\mu=\lambda(\xi+\eta)$ and $\nu=\lambda\xi\eta$. Plugging the latter result along with (\ref{Delta}) into (\ref{proof aux}), we finally arrive at
\[
\frac{1}{\hat{a}_2-\eta\hat{b}_2} = \frac{\lambda}{\nu}(a_1+\xi b_1+a_3+\xi b_3)+
\frac{\lambda}{\nu}\frac{(a_1+\xi b_1)(a_3+\xi b_3)}{a_2+\xi b_2},
\]
which demonstrates (one of) equations (\ref{coupled star triang}).
\qed

{\bf Remark.} The oscillating sign $(-1)^{|n|}$ in (\ref{main cond}) can be explained in the following way. The matrices $A^{ij}$ in (\ref{matr Moutard}) can be interpreted as mapping from the diagonal $(n+e_i,n+e_j)$ to the diagonal $(n,n+e_i+e_j)$ of an elementary square $\sigma^{ij}$. A more geometric way would be to consider matrices mapping black diagonals (connecting points with $|n|$ even) to white diagonals (connecting points with $|n|$ odd). These are matrices
 \[
 B^{ij}=\big(A^{ij}\big)^{(-1)^{|n|}}.
 \]
 For these matrices, constraint (\ref{lin dep}) holds on all faces of $\bbZ^N$, without changing the sign of $\mu$.
 \medskip

 {\bf Example 1.} $\mu=0$, $\lambda/\nu=1$, so that $a^{ij}=-c^{ij}$, and $\{\xi,\eta\}=\{\pm \I\}$. Since $\xi=-\eta$, it is more convenient to replace definitions (\ref{pq}), (\ref{pq hat}) by the ``non-oscillating'' ones, namely,
 \[
 p^{ij}=a^{ij}+\I b^{ij}, \quad q^{ij}=a^{ij}-\I b^{ij}=\bar{p}^{ij}.
 \]
Thus, we can set
\[
 a^{ij}=-c^{ij}=\Re(p^{ij}),\quad b^{ij}=\Im(p^{ij}),
\]
where the complex-valued coefficients $p^{ij}=-p^{ji}$ evolve according to the star-triangle equation:
\begin{equation}\label{B.p}
 \frac{1}{p^{ij}_k}=\frac{p^{ij}p^{jk}+p^{jk}p^{ki}+p^{ki}p^{ij}}{p^{ij}}.
\end{equation}

 {\bf Example 2.} $\mu=0$, $\lambda/\nu=-1$, so that $a^{ij}=c^{ij}$, and $\{\xi,\eta\}=\{\pm 1\}$. Like in the previous example, we replace definitions (\ref{pq}), (\ref{pq hat}) by the ``non-oscillating'' ones, namely,
 \[
 p^{ij}=a^{ij}+b^{ij}, \quad q^{ij}=a^{ij}-b^{ij}.
 \]
Thus, we can set
\[
 a^{ij}=c^{ij}=\frac{1}{2}(p^{ij}+q^{ij}),\quad b^{ij}=\frac{1}{2}(p^{ij}-q^{ij}),
\]
where the coefficients $p^{ij}=-p^{ji}$, $q^{ij}=-q^{ji}$ evolve according to two (uncoupled) star-triangle equations:
\begin{equation}\label{B.ab}
 \frac{1}{p^{ij}_k}=-\frac{p^{ij}p^{jk}+p^{jk}p^{ki}+p^{ki}p^{ij}}{p^{ij}},\qquad
 \frac{1}{q^{ij}_k}=-\frac{q^{ij}q^{jk}+q^{jk}q^{ki}+q^{ki}q^{ij}}{q^{ij}}.
\end{equation}

{\bf Example 3.} $\mu=0$, $\lambda=0$, thus $c^{ij}=0$. The parametrization of Theorem \ref{th star-triang classif}, taken literally, does not work, but the matrix star-triangle relations (\ref{matr star triang}) for triangular matrices $A^{ij}$ can be easily solved. Indeed, from (\ref{matr star triang elem}) with $c^{ij}=0$ we find  directly:
\begin{eqnarray}
 -\dfrac{1}{b^{ij}_k} & = & b^{jk}+b^{ki}+\dfrac{b^{jk}b^{ki}}{b^{ij}}, \label{ex3: b}\\
 \dfrac{a^{ij}_k}{(b^{ij}_k)^2} & = & a^{jk}+a^{ki}
  +\dfrac{b^{ki}a^{jk}}{b^{ij}}+\dfrac{b^{jk}a^{ki}}{b^{ij}}
  -\dfrac{b^{jk}b^{ki}a^{ij}}{(b^{ij})^2}. \label{ex3: a}
\end{eqnarray}
Thus, evolution of the entries $b^{ij}$ is governed by the star-triangle relation, while the equation satisfied by the entries $a^{ij}$ can be interpreted as the linearized equation for $b^{ij}$.
\medskip

{\bf Example 4.} $\nu=0$, $\lambda=0$, thus $b^{ij}=0$. Again, the parametrization of Theorem \ref{th star-triang classif}, taken literally, does not work, but the matrix star-triangle relations (\ref{matr star triang}) for off-diagonal matrices $A^{ij}$ can be easily solved directly: from (\ref{matr star triang elem}) with $b^{ij}=0$ we find:
\begin{eqnarray}
 -\frac{1}{c^{ij}_k} & = & a^{jk}+a^{ki}+\frac{a^{jk}a^{ki}}{a^{ij}}, \label{ex4: c}\\
 -\frac{1}{a^{ij}_k} & = & c^{jk}+c^{ki}+\frac{c^{jk}c^{ki}}{c^{ij}}, \label{ex4: a}
\end{eqnarray}
so that the entries $a^{ij}$ and $c^{ij}$  evolve according to the coupled star-triangle equations. The admissible reduction $a^{ij}=-c^{ij}$ is also the special case $p^{ij}=a^{ij}\in \bbR$ of Example 1.

\subsection{From vector Moutard equations to discrete pluriharmonic functions}

Now, we can translate the results of the previous section into the language of discrete pluriharmonic functions.

\begin{theorem}
Lagrangians
\[
L^{ij}=\frac{1}{2}\alpha^{ij}(u_{ij}-u)^2+\beta^{ij}(u_{ij}-u)(u_i-u_j)+
\frac{1}{2}\gamma^{ij}(u_i-u_j)^2.
\]
assigned to oriented plaquettes $\sigma^{ij}(n)=(n,n+e_i,n+e_i+e_j,n+e_j)$ in $\mathbb Z^N$ are consistent, i.e. form a pluri-Lagrangian system, if and only if their coefficients satisfy
\begin{equation}\label{main cond}
\lambda-\mu(-1)^{|n|}\beta^{ij}+\nu((\beta^{ij})^2-\alpha^{ij}\gamma^{ij})=0
\end{equation}
for some fixed triple $(\lambda,\mu,\nu)$, where $|n|=n_1+\ldots+n_m$,  and their evolution is expressed (in the generic case, when $\lambda,\nu\neq 0$) through a solution of {\em coupled star-triangle relations} (\ref{coupled star triang}) via the following relations:
\begin{equation}\label{pq Lagr}
p^{ij}=\frac{1-\xi\beta^{ij}}{\alpha^{ij}},\quad q^{ij}=\frac{1-\eta\beta^{ij}}{\alpha^{ij}},
\end{equation}
and
\begin{equation}\label{pq hat Lagr}
p^{ij}_k=\frac{1+\xi\beta^{ij}_k}{\alpha^{ij}_k},\quad q^{ij}_k=\frac{1+\eta\beta^{ij}_k}{\alpha^{ij}_k},
\end{equation}
where $-\xi$, $-\eta$ are the two roots of the quadratic equation $\lambda \xi^2+\mu\xi+\nu=0$.
\end{theorem}
{\bf Example 1.} $\mu=0$, $\lambda/\nu>0$, without loss of generality one can normalize
\[
(\beta^{ij})^2-\alpha^{ij}\gamma^{ij}=-1.
\]
The Lagrangians can be parametrized as
\[
\alpha^{ij}=\frac{1}{\Re(p^{ij})},\quad \beta^{ij}=\frac{\Im(p^{ij})}{\Re(p^{ij})},\quad
\gamma^{ij}=\frac{|p^{ij}|^2}{\Re(p^{ij})},
\]
where $p^{ij}$ is a complex-valued solution of the star-triangle relation (\ref{B.p}). This is the case of discrete complex analysis considered in Section \ref{s:H}.
\medskip

{\bf Example 2.} $\mu=0$, $\lambda/\nu<0$, without loss of generality one can normalize
\[
(\beta^{ij})^2-\alpha^{ij}\gamma^{ij}=1.
\]
The Lagrangians can be parametrized as follows:
\[
\alpha^{ij}=\frac{2}{p^{ij}+q^{ij}},\quad \beta^{ij}=\frac{p^{ij}-q^{ij}}{p^{ij}+q^{ij}}, \quad
\gamma^{ij}=-\frac{2p^{ij}q^{ij}}{p^{ij}+q^{ij}},
\]
so that
\begin{equation}\label{B.L-ab}
 L^{ij}=\frac{1}{p^{ij}+q^{ij}}
  (u_{ij}-u+p^{ij}(u_i-u_j))(u_{ij}-u-q^{ij}(u_i-u_j)),
\end{equation}
where the coefficients $p^{ij}=-p^{ji}$, $q^{ij}=-q^{ji}$ evolve according to two star-triangle equations (\ref{B.ab}).
\medskip

{\bf Example 3.} $\mu=0$, $\lambda=0$, thus $(\beta^{ij})^2-\alpha^{ij}\gamma^{ij}=0$. The Lagrangians can be parametrized as
\begin{equation}\label{B.L-as}
 L^{ij}=\frac{1}{a^{ij}}(u_{ij}-u+b^{ij}(u_i-u_j))^2,
\end{equation}
where the evolution of the coefficients $b^{ij}=-b^{ji}$ and $a^{ij}=-a^{ji}$ is governed by equations (\ref{ex3: b}), (\ref{ex3: a}).
\medskip

{\bf Example 4.} $\nu=0$, $\lambda=0$, thus $\beta^{ij}=0$. The Lagrangians can be parametrized as
\begin{equation}\label{B.L-double}
 L^{ij}=\frac{1}{a^{ij}}(u_{ij}-u)^2-c^{ij}(u_i-u_j)^2,
\end{equation}
where the coefficients $a^{ij}=-a^{ji}$, $c^{ij}=-c^{ji}$ evolve according to the coupled star-triangle equations (\ref{ex4: c}), (\ref{ex4: a}). The admissible reduction $a^{ij}=-c^{ij}$ is the special case $p^{ij}\in \bbR$ of Example 1. In the complex analysis interpretation this is the case of quads with orthogonal diagonals \cite{M_CMP, Skopenkov}.

\section{Pluri-Lagrangian systems with general quadratic Lagrangians}\label{s:Qnet}

Diagonal Lagrangians do not exhaust the class of linear pluri-Lagrangian systems. To show this, we now give the probably simplest example of such a system, whose Lagrangians belong to the so called three-point class.

\begin{theorem}\label{th:BKP-a}
Lagrangians
\begin{equation}\label{B.L-a}
 L^{ij}=(u_i-u)^2-(u_j-u)^2-p^{ij}(u_i-u_j)^2
\end{equation}
are consistent if and only if parameters $p^{ij}=-p^{ji}$ evolve according to the star-triangle relation
\begin{equation}\label{B.a}
 p^{ij}_k=-\frac{p^{ij}}{p^{ij}p^{jk}+p^{jk}p^{ki}+p^{ki}p^{ij}}.
\end{equation}
\end{theorem}
Note that Definition \ref{def: corner eqs consist} of consistency of corner equations reads in the case of three-point 2-forms literally the same, but its interpretation should be slightly modified. The quantity $S^{ijk}$ defined as in (\ref{eq: Sijk}) does not depend on $u$, $u_{ijk}$. Thus, system of corner equations (\ref{eq: corner eqs}) consists of 6 equations, each of them involving 5 variables. Initial data consist of four values (say, $u_i$, $u_j$, $u_k$, $u_{ij}$), then two of the corner equations determine the remaining two values ($u_{jk}$ and $u_{ik}$ in the above example), and then the remaining four corner equations should be fulfilled automatically.

{\em Proof of Theorem \ref{th:BKP-a}.}
Corner equations for the Lagrangians (\ref{B.L-a}) read as follows:
\begin{eqnarray}
 && u_{ij}-u_{ik}-p^{ij}(u_i-u_j)+p^{ki}(u_k-u_i)=0, \label{Ei}\\
 && u_i-u_j+p^{jk}_i(u_{ij}-u_{ik})-p^{ki}_j(u_{jk}-u_{ij})=0. \label{Eij}
\end{eqnarray}
A direct check shows that two of equations (\ref{Ei}) imply the third one, while equations (\ref{Eij}) with $u_{jk}$ and $u_{ik}$ determined from (\ref{Ei}) are equivalent to
\[
p^{ij}_k p^{ki}=p^{ki}_j p^{ij},\quad p^{ij}_k(p^{jk}+p^{ki})+p^{ki}_jp^{jk}=-1.
\]
These equations are easily solvable, with the solution given by (\ref{B.a}). \qed


We now give an example of  a pluri-Lagrangian system with 3 parameters per face.

\begin{theorem}\label{th:Q}
Lagrangians
\begin{equation}\label{Q.L}
 L^{ij}=\frac{1}{2s^{ij}c^{ij}c^{ji}}
  (u_{ij}-c^{ji}u_i-c^{ij}u_j-c^{ij}c^{ji}u)^2,
 \quad s^{ij}=-s^{ji}
\end{equation}
are consistent in the sense of the definitions \ref{def:L-consistency'}, \ref{def: corner eqs consist}, if and only if the coefficients are governed by the following (4D-consistent) mapping:
\begin{eqnarray}
 c^{ij}_k & = & \frac{1}{c^{kj}}(c^{ik}c^{ki}-c^{ik}c^{kj}-c^{ki}c^{ij}), \label{Q.c}\\
 s^{ij}_k & = & c^{ki}c^{kj}s^{ij}
     +c^{ki}(c^{ij}-c^{ik})s^{jk}
     +c^{kj}(c^{ji}-c^{jk})s^{ki}. \label{Q.s}
\end{eqnarray}
\end{theorem}

\begin{proof}
We present a computer algebra proof by Vsevolod Adler in \linebreak
{\tt page.math.tu-berlin.de/$\sim$bobenko}
\end{proof}

The system in Theorem \ref{th:Q} is related to discrete conjugated nets (Q-nets).
Recall that a Q-net is a multidimensional lattice with planar
quadrilaterals \cite{DDG}. An obvious analytic description of such nets, which clearly belong to projective geometry, is given in the homogeneous coordinates by the system of linear equations
\begin{equation}\label{Q.u-p}
 u_{ij}=c^{ji}u_i+c^{ij}u_j+d^{ij}u,\quad d^{ij}=d^{ji}.
\end{equation}
A convenient  gauge (choice of the representative of homogeneous coordinates) can be achieved as follows. The compatibility conditions of (\ref{Q.u-p}) contain the conservation law
\[
 \frac{r^{ij}_k}{r^{ij}}=
 \frac{r^{jk}_i}{r^{jk}}=
 \frac{r^{ki}_j}{r^{ki}},\quad r^{ij}=\frac{c^{ij}c^{ji}}{d^{ij}}
\]
which allows to introduce a function $g$ such that
\[
 \frac{c^{ij}c^{ji}}{d^{ij}}=\frac{gg_{ij}}{g_ig_j}.
\]
Now, the substitution $\tilde u=u/g$ leads to equation
\begin{equation}\label{Q.u}
 u_{ij}=c^{ji}u_i+c^{ij}u_j+c^{ij}c^{ji}u.
\end{equation}
This gauge was introduced in \cite{Sergeev2000a,Sergeev2000b}.  The compatibility condition for the linear system (\ref{Q.u}) is given by  (\ref{Q.c}). Corner equations for the Lagrangians (\ref{Q.L}) are corollaries of (\ref{Q.u}).
The consistency conditions of corner equations include the consistency equations (\ref{Q.c}) for (\ref{Q.u}).

\section{Concluding remarks}\label{s:conclusion}
A classification of general linear pluri-Lagrangian systems remains an interesting open problem.
Such systems may lead to new 3D discrete integrable systems for the coefficients generalizing the star-triangle equation.

An interesting open problem is to develop an analysis of discrete harmonic and discrete
holomorphic functions on quad-graphs with arbitrary quadrilaterals based on the
extension to pluriharmonic functions.  For example, one can hope to obtain
explicit formulas for $\exp,\log$ and Green's functions like in the case of
rhombic (isoradial) graphs \cite{Ke, BMS, ChS}. The difference is that in the case of arbitrary quadrilaterals
we are dealing with a 3 dimensional integrable system in contrast to the isoradial case where
the corresponding integrable system is 2 dimensional. Unfortunately analytic methods for 3 dimensional integrable systems are not yet that well developed as for 2 dimensional ones.

Let us also mention that in all cases found by us the Lagrangians $L^{ij}$ of pluri-Lagrangian linear systems admit a factorization into linear factors. The role of this property is not clear at the moment.


\section*{Acknowledgements}

We are grateful to Vsevolod Adler for collaboration, in particular for computer algebra investigation of linear pluri-Lagrangian systems.
Research for this article was supported by the SFB/TR 109 ``Discretization in
Geometry and Dynamics''.

\phantomsection \addcontentsline{toc}{section}{References}

\bibliographystyle{amsalpha}

\end{document}